\documentclass[letterpaper, 10 pt, conference]{ieeeconf}  

\IEEEoverridecommandlockouts                              

\renewcommand{\theenumii}{\theenumi.\arabic{enumii}.}
\usepackage{amsmath}
\usepackage{amsfonts}
\usepackage{amssymb}
\usepackage{graphicx}
\usepackage{MnSymbol,wasysym}
\usepackage{cancel}
\usepackage{hyperref}
\usepackage{authblk}

 \newtheorem{theorem}{\textit{Theorem}}
 
 \newtheorem{lemma}{\textit{Lemma}}

\usepackage{gensymb}

\date{\today}
\title{\LARGE Local Interactions for Cohesive Flexible Swarms}
\author{Rotem Manor$^{1}$, Ariel Barel$^{2}$, and Alfred M. Bruckstein$^{3}$}
\affil{Center for Intelligent Systems (CIS)\\ Computer Science Department \\ Technion, Haifa 32000, Israel.}

\begin{document}
\maketitle
\begin{abstract}

Distributed gathering algorithms aim to achieve complete visibility graphs via a ``never lose a neighbour" policy. We suggest a method to maintain connected graph topologies, while reducing the number of effective edges in the graph to order $n$. This allows to achieve different goals and swarming behaviours: the system remains connected but flexible, hence can maneuver in environments that are replete with obstacles and narrow passages, etc.\\

\end{abstract}

\addcontentsline{toc}{section}{Introduction}

\section*{Introduction}

We present a novel method for distributed control of multi-agent systems, which ensures maintaining flexible and connected spatial constellations. Such a requirement is needed in the context of many tasks that swarms of agents must carry out, such as mapping unknown environments, and deploying a connected network of sensors to cover unknown regions. The agents are assumed to be identical, anonymous (i.e. indistinguishable) and simple in the sense of having little or no memory (i.e. oblivious), with limited computation and sensing capabilities. In this paper we assume limited sensing range, denoted by $V$. The swarm remains connected, never splits into disjoint groups, hence remains able to perform tasks cohesively.\\

We consider systems comprising mobile agents that interact solely by adjusting their motion according to the relative location of their neighbours. The agents are assumed capable of sensing the presence of other agents within the given sensing range. The agents then implement rules of motion based on information on the geometric constellation of their neighbours.\\

The motion of agents is designed to ensure global connectivity without attempting to maintain all existing visibility connections between agents. Therefore, we obtain swarms that are cohesive but not rigid, allowing the agents to move more freely and assume various desirable formations, for example in order to pass through narrow passages between obstacles, and reduce communication loads when agents must act as a backbone network for communication, etc. Our algorithm is truly distributed, hence insensitive to the size of the swarm, since each agent carries out simple calculations and makes local individual decisions, rendering it suitable for swarms with very large numbers of agents.\\

Interactions between agents in multi-agent systems are often mathematically described using a graph, commonly labelled as $G(\mathcal{V},\mathcal{E})$, where $\mathcal{V}=\{\nu_1, \nu_2,...,\nu_n\}$ is the set of vertices (representing the agents), and $\mathcal{E} \subseteq \mathcal{V} \times \mathcal{V}$ is the set of edges (representing connections between the agents). The neighbourhood set of a vertex is the set of vertices connected to it, i.e.  $N_i \triangleq \{\nu_j\in \mathcal{V}\; |\; \{\nu_i,\nu_j \} \in \mathcal{E} \}$.\\

The connection between the agents is either fixed and known in advance, or it is defined based on geometric relationships between the agents' positions (in space, say in the plane $\mathbb{R}^2$), i.e. on the set $\{p_i \triangleq (x_i,y_i)\}_{i=1, 2, 3, ...}$. In the limited visibility case, agents are visible to each other if their mutual distance does not exceed the visibility range limit $V$, and the interconnection graph $G(\mathcal{V},\mathcal{E})$ is the visibility graph corresponding to the agents' locations. In this case, in $G(\mathcal{V},\mathcal{E})$, the set $\mathcal{E}$ is defined as follows:
$$
e_{i,j} \in \mathcal{E} \iff \|p_i-p_j\| \leq V
$$
A fundamental issue in a multi-agent system is to maintain cohesiveness, i.e. the interconnection graph corresponding to the system configuration, must stay connected, otherwise the system splits into disjoint parts. For a graph to be connected, there must always be a path between each pair of nodes in the graph.\\

In a recent survey on multi-agent geometric consensus \cite{barel2019come}, different types of distributed gathering algorithms are presented in detail, where gathering agents with limited visibility is addressed, as discussed in \cite{ando1999,gordon2004,gordon2005,gazi2003stability,gazi2004stability,olfati2007consensus,moreau2004,jadbabaie2003}. Other recent articles e.g. \cite{ozdemir2018finding}, \cite{dimarogonas2007rendezvous}, \cite{su2010rendezvous}, \cite{barel2019probabilistic}, \cite{manor2016chase} also deal with the consensus problem in distributed multi-agent systems. A common principle in the distributed control algorithms dealing with gathering, rendezvous, clustering, and aggregation of multi-agent systems is the "never lose a neighbour" policy, requiring:
$$
\forall i \;\: \text{and} \;\: \Delta t \geq 0 \quad : \quad j \in N_{i}(t) \Rightarrow j \in N_{i}(t+\Delta t)
$$
that is, once agents become neighbours, they must remain neighbours forever, hence the number of edges in the interconnection graph never decreases. Local behaviours for the agents that enforce this policy, with additional bias aiming to gradually add more edges to the graph, are often sufficient for the graph to become complete in finite expected time. In the limited visibility case, such systems necessarily gather to a region with ``diameter" less than $V$.\\

Hence, the ``never lose a neighbour" policy is optimally suited for the purpose of \textit{gathering}, since increasing the number of connected agents leads to clustering of agents to a small, visibility-horizon-defined area. However, for connectivity maintenance only, the ``never lose a neighbour" policy is certainly too conservative and stringent: maintaining connectivity only requires a \textit{connected} neighbourhood graph, rather than a \textit{complete} graph, and this goal may be achieved if some existing mutually visible agents cease to be defined as neighbours.\\

The conservative policy is quite popular in the multi-agent distributed control field since the lack of coordination between agents' decisions, may, in general, cause the graph to become disconnected. We here propose a local method for maintaining a connected graph in distributed multi-agent systems, which ensures, along with preserving the graph connectivity, that the number of edges does not exceed $3n$, where $n=| \mathcal{V} |$ is the number of agents in the system. The proposed method maintains graph connectivity with $| \mathcal{E} | \leq 3n$ as compared to $| \mathcal{E} | = \frac{n(n-1)}{2}$ in the complete graph.\\

Why can such a generic method be useful? A swarm in which the constraints that apply to agents are fewer is more flexible, so it can be steered through regions full of obstacles and narrow passages. A swarm that is connected but the number of neighbourhood defining edges in the graph is small, may self-organise in more interesting and useful constellations. A swarm in which the connections are determined according to the proposed algorithm can be spread and stretched so that the distance between its agents can attain distance of the order of $(n-1)V$, i.e. in some cases it can even be arranged to form a line. Hence we envisage that the proposed method can be a useful basic process in designing multi-agent swarming algorithms that are usually called upon to ensure cohesiveness in various cases of controlled swarm actions, where, for example, the swarm is guided by leaders or by some exogenous broadcast controls.\\

\section*{The cohesion ensuring process}
Following Toussaint's work (1980) on relative neighbourhood graphs (RNG) \cite{toussaint1980relative}, we define $N^{e}_{i}$, the ``effective neighbourhood" of agent $i$, as a subset of $N_{i}$:\\

\begin{equation}
\forall j \in N_{i} \\\ : \\\ j \in N^{e}_{i} \iff \nexists k : \left\{\begin{alignedat}{3}
    & \|p_{i}-p_{k}\| < \|p_{i}-p_{j}\| \\
    & \& \\
    & \|p_{j}-p_{k}\| < \|p_{i}-p_{j}\|\\
  \end{alignedat}\right.
\label{eq:EffectiveN}
\end{equation}
i.e. if agents $i$ and $j$ are neighbours, and there is an agent $k$ whose distance from agent $i$ and from agent $j$ is smaller than the distance between agents $i$ and $j$ (see Figure \ref{EdgeTrimming}), then the edge $e_{ij}$ is not necessary to maintain connectivity between agents $i$ and $j$. Coordination between the agents is already a feature built into this method because the decision of two agents whose connection is not necessary to maintain connectivity is symmetrical, that is, if a ``bridging" agent $k$ exists, we have that $j \notin N^{e}_{i} \iff i \notin N^{e}_{j}$.\\

\begin{figure}[ht]
  \centering
    \includegraphics[width=55mm]{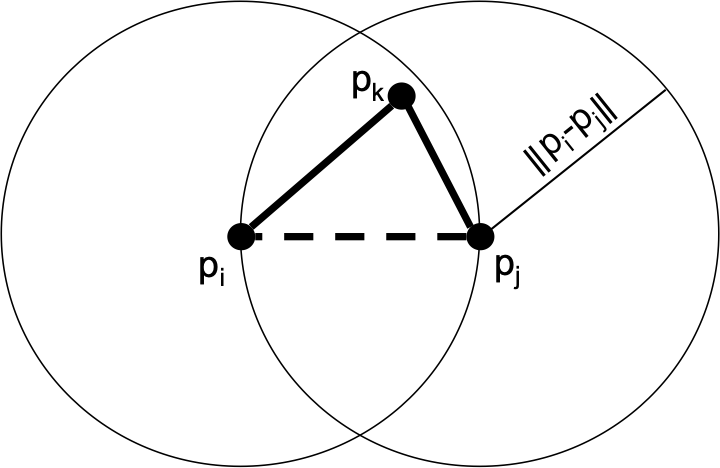}
    \caption{Two circles of radius $\|p_{i}-p_{j}\|<V$ centered at $p_{i}$ and $p_{j}$. If there is an agent $k$ inside the intersection lens of these circles, it is necessarily closer to $i$ than $j$ and closer to $j$ than $i$. Hence in this case $j \in N_{i}$ but $j \notin N^e_{i}$, and $i \in N_{j}$ but $i \notin N^e_{j}$.}
      \label{EdgeTrimming}
\end{figure}

Define $G^e$ as $G(\mathcal{V},\mathcal{E}^e)$, i.e. $G^e$ is a subgraph of the visibility graph $G$ whose set of edges $\mathcal{E}^e$ is defined by (\ref{eq:EffectiveN}).\\

An example of a neighbourhood graph vs. effective neighbourhood graph is shown in Figure \ref{Mesh}, where both solid and dashed lines belong to the neighbourhood graph, but only the solid lines belong to the effective neighbourhood graph. In this example only half of the edges remain after reducing unnecessary edges due to rule (\ref{eq:EffectiveN}).   

\begin{figure}[ht]
  \centering
    \includegraphics[width=80mm]{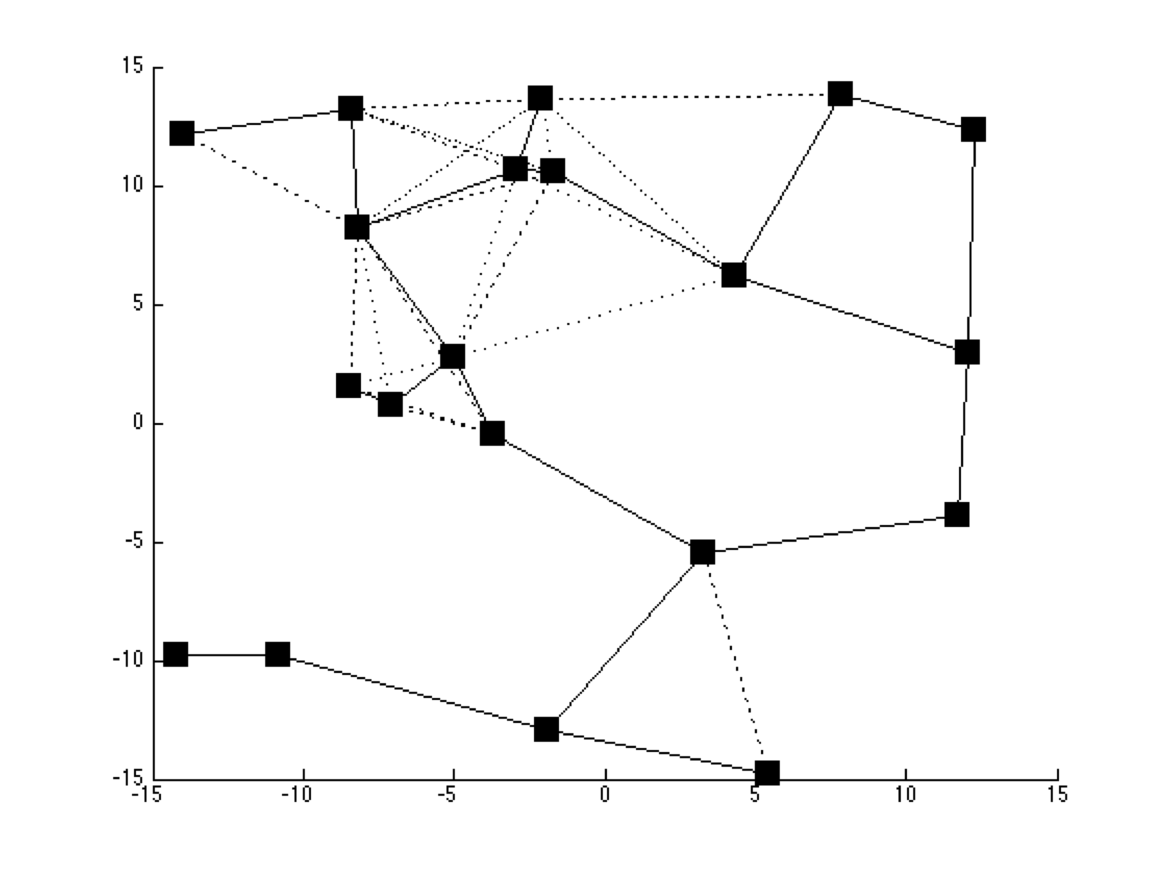}
    \caption{In this example, the original neighbourhood graph on $20$ nodes, contain $42$ edges between every pair of agents distanced $V$ or less (both solid and dashed lines). In the effective neighbourhood graph, $21$ edges that are unnecessary for connectivity maintenance due to rule (\ref{eq:EffectiveN}) are shown dashed, and only $21$ edges, shown in solid line, remain.}
      \label{Mesh}
\end{figure}

\begin{lemma}\label{EffectiveN_Connected}
If $G$ is connected, then $G^e$ is connected.
\end{lemma}

\begin{proof}
We prove this Lemma by contradiction. Assume $G$ is connected and $G^e$ is disconnected to two disjoint connected sets $\mathcal{V}_{1}$ and $\mathcal{V}_{2}$. Hence, by assumption,
\begin{equation}
\forall (i \in \mathcal{V}_{1} , j \in \mathcal{V}_{2}) : \exists k : \left\{\begin{alignedat}{3}
    & \|p_{i}-p_{k}\| < \|p_{i}-p_{j}\| \\
    & \& \\
    & \|p_{j}-p_{k}\| < \|p_{i}-p_{j}\|\\
  \end{alignedat}\right.
\label{Contradiction}
\end{equation}
without loss of generality, assume $i \in \mathcal{V}_{1}$ and $j \in \mathcal{V}_{2}$ are the closest agents between the two groups $\mathcal{V}_{1}$ and $\mathcal{V}_{2}$. Since $G$ is connected by assumption, we have that $\|p_i-p_j\| \leq V$, being the closest pair. By assumption these two agents are not effective neighbours, hence are not connected in $G^e$, i.e. $i \notin N^e_{j}$ and $j \notin N^e_{i}$. But by (\ref{eq:EffectiveN}) there is necessarily an agent $k$ that links them, see (\ref{Contradiction}). Hence $i$ and $j$ can not be the closest agents.
\end{proof}

We note that the RNG has been proposed in multi-agent robotics in conjunction with multi-agent gathering and/or deployment algorithms, e.g. \cite{li2017continuous}, \cite{ganguli2005collective}, \cite{cortes2006robust}, \cite{martinez2007motion}. The novelty in our work is $(1)$ proposing the use of the RNG and some extensions for maximally flexible swarm operations, and $(2)$ defining and using, along with various goal oriented constraints, the largest local allowable regions for displacement that ensure maintenance of connectivity, in the spirit of \cite{gordon2005}. We note that methods to reduce the number of edges in geometric graphs while maintaining connectivity are also used for purposes that do not concern dynamics of agents, but, for example, for saving energy in ad-hoc communication networks, see e.g. \cite{koushanfar2007techniques}.

\section*{The motion Laws}

In order to allow the agents free movement that does not violate the preservation of connectivity, we define for each agent an ``Allowable Region". If all agents move into their allowable regions, preservation of connectivity is satisfied. Recalling Ando et. al. \cite{ando1999} for a pair of agents, the Allowable Region of each agent is a disc of radius $\frac{V}{2}$ centered at the midpoint position between the two agents $m_{ij}$ (see figure \ref{AllowableRegion2}).

\begin{equation}
AR_{ij}=AR_{ji}=D_{\frac{V}{2}}\left(\frac{p_i+p_j}{2}\right)
\label{AR}
\end{equation}

\begin{figure}[ht]
  \centering
    \includegraphics[width=40mm]{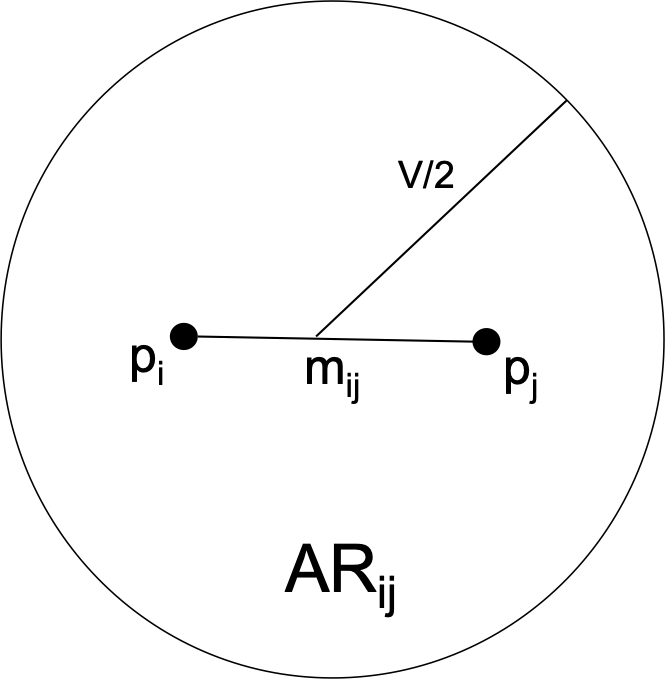}
    \caption{The Allowable Region $AR_{ij}$ of two agents $i$ and $j$ is a disc of radius $\frac{V}{2}$ centered at their meadpoint position $m_{ij}$.}
      \label{AllowableRegion2}
\end{figure}

If two agents whose mutual distance is smaller or equal $V$ move into their allowable region defined in (\ref{AR}) that they both can calculate, they keep their mutual distance smaller than $V$, hence the agents will stay connected.

If an agent has more than one neighbour, its allowable region is defined by the union of is allowable regions given by (\ref{AR}), see Figure \ref{AllowableRegionMany}:

\begin{equation}
AR_i=\bigcap\limits_{j\in N_{i}}AR_{ij}
\label{ARs}
\end{equation}

\begin{figure}[ht]
  \centering
    \includegraphics[width=55mm]{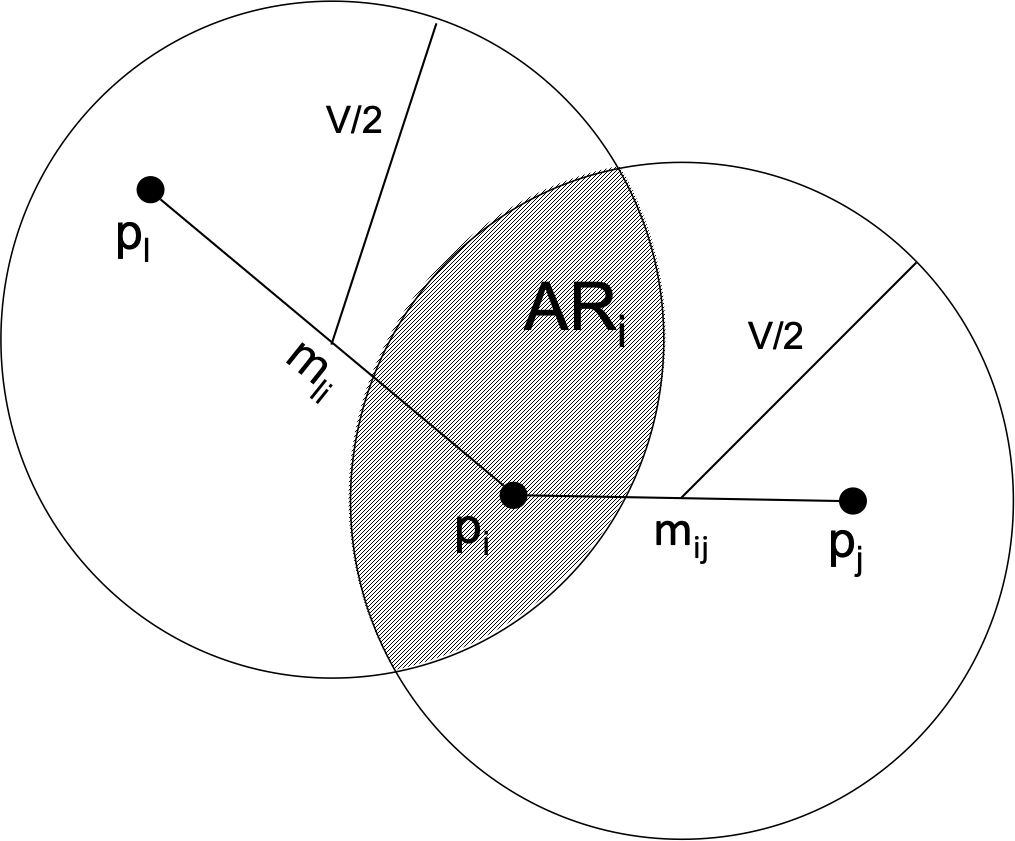}
    \caption{The allowable region of an agent is the union of its allowable regions, illustrated here by the dashed area $AR_i$.}
      \label{AllowableRegionMany}
\end{figure}

clearly if an agent moves to the union of its allowable regions defined by (\ref{ARs}), it is guaranteed that the distance from all its neighbours will not exceed $V$.\\

Similar to (\ref{ARs}),we define an ``Effective Allowable Region":

\begin{equation}
AR^e_i=\bigcap\limits_{j\in N^e_{i}}AR_{ij}
\label{EfARs}
\end{equation}
i.e. the effective allowable region of agent $i$ is the intersection of all its allowable regions relative to all the agents in its \textit{effective} neighbourhood, defined by $G^e$.\\

Hence, the motion law will be:
\begin{equation}
p_i(t) \longrightarrow p_i(t+1) \in AR^e_i(t)
\label{MotionLaw}
\end{equation}
i.e. the next position of agent $i$ is \textit{any point} inside its current effective allowable region (\ref{EfARs}). Since $AR_i(t) \subseteq AR^e_i(t)$, the restriction of movement in (\ref{EfARs}) is less constrained than that of (\ref{ARs}).

The freedom to move to any location inside allowable regions rather to a specific point, provides flexibility and allows other considerations and optimisations to be carried out by the agents.\\

\begin{theorem}
A system of agents obeying dynamic law (\ref{MotionLaw}) maintains connectivity.
\end{theorem}

\begin{proof}
The initial visibility graph $G(\mathcal{V},\mathcal{E})(0)$ associated with the system is connected by assumption. We have by Lemma \ref{EffectiveN_Connected} that $G^e(\mathcal{V},\mathcal{E})(t)$ is connected. If all agents obey motion law (\ref{MotionLaw}), we have that $\mathcal{E}^e(t) \subseteq \mathcal{E}(t+1)$, i.e. all edges of $G^e(\mathcal{V},\mathcal{E})$ are preserved. Therefore $G(\mathcal{V},\mathcal{E})(t+1)$ is connected. Since $G(\mathcal{V},\mathcal{E})(t+1)$ is connected, also $G^e(\mathcal{V},\mathcal{E})(t+1)$ will be connected.
\end{proof}

\section*{Applications and Simulation Results}

The distributed meta-algorithm presented so far opens up many possibilities of applications for distributed control of swarms. A number of interesting research directions that we consider relevant are algorithms with the ability to move a swarm of agents through narrow passages and obstacles, algorithms for mapping unknown areas, and reducing communication load when agents constitute a relay station for transmitting messages, as well as other topics requiring geometrically flexible swarms. We discuss below several examples.\\

\textbf{Example $1$: Traversal of narrow passages}\\

In this example (Figure \ref{NarrowPassageway}), which will be discussed in detail in a subsequent article, the swarm is required to pass through a narrow passage. we implemented a rule of motion that prevents two agents from approaching each other closer than $\frac{1}{10}V$. Then the agents movement is also governed by the cohesion rule  (\ref{MotionLaw}) with two additional restrictions on the agents' allowable regions: (1) they exclude the obstacle areas (1) they ensure line-of-sight preservation between effective neighbours. One designated ``leader" agent pulls the rest of the agents through the narrow passage. Because the swarm is not rigid, this passage is possible. By comparison, if the ``never lose a neighbour" policy would have been applied together with the minimum distance limit ($\frac{1}{10}V$), the swarm would have been too rigid to pass through the narrow passage.

\begin{figure}[ht]
  \centering
    \includegraphics[width=89mm]{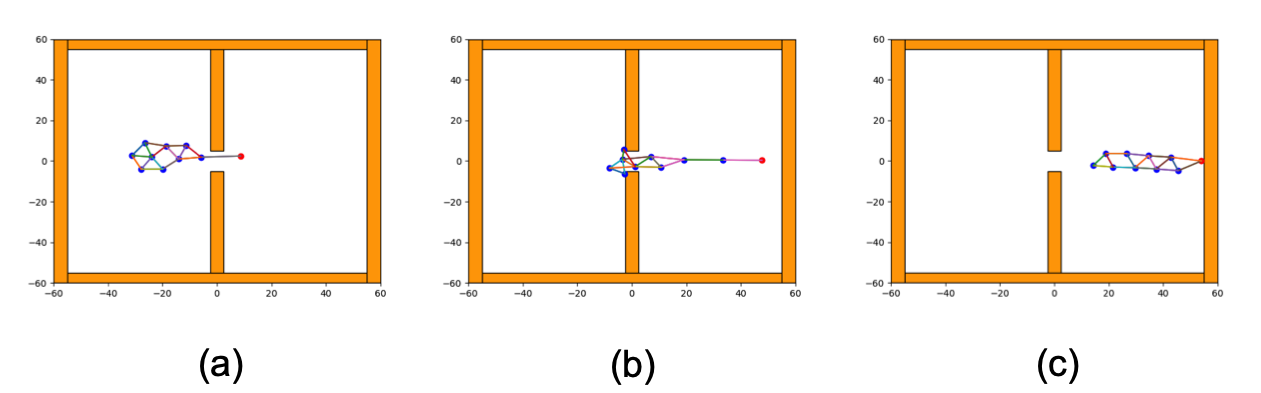}
    \caption{Passing through a narrow passage. (a) Initial constellation (b) The leader pulls the rest of the agents while only necessary edges are preserved (c) The swarm successfully completed the passage. This simulation result is based on the ``RNG Plus" approach described at the end of this article with the value $m=1$.}
      \label{NarrowPassageway}
\end{figure}

\textbf{Example $2$: Formation}\\

This example presents simulations of convergence to quasi regular formation (Figure \ref{StructuredMovement}). In a topology where each agent relates to all other agents (i.e. in complete graph), by geometry it is impossible to create a structure in which all the lengths of the edges are identical. In this simulation, each agent worked according to rule (\ref{eq:EffectiveN}) in order to choose only its effective neighbours, and only in relation to these (effective) neighbours he aspires to maintain a fixed distance of $\frac{1}{10}V$. The system ``anneals" to a nice quasi regular formation enabling nice and uniform area coverage as presented in the network evolution below.\\

\begin{figure}[ht]
  \centering
    \includegraphics[width=89mm]{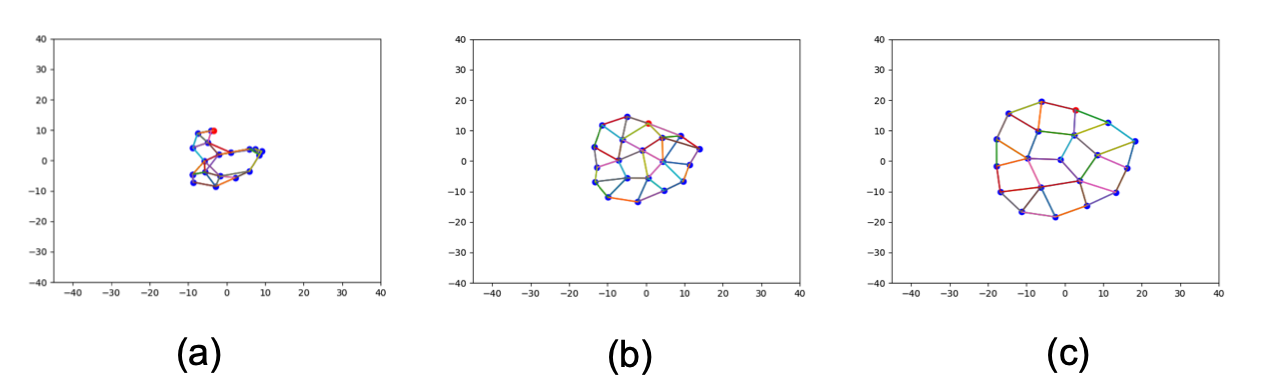}
    \caption{(a) Edges exist between every pair of mutually visible agents at a distance $V$ or less (b) Evolution of the network: only the necessary edges are kept (c) Final constellation: a required structured is achieved. This simulation result is based on the ``RNG Plus" approach described at the end of this article with the value $m=1$.}
      \label{StructuredMovement}
\end{figure}

\textbf{Example $3$: Following a leader}\\

In this example (Figure \ref{Leader}), one random agent in the group becomes a designated leader which will seek to move along a predetermined path. A process is shown in which the number of edges in the neighbourhood graph decreases, while the graph remains connected, until a linear constellation of agents is formed, and the linear ``snake-like" swarm follows the leader.

\begin{figure}[ht]
  \centering
    \includegraphics[width=89mm]{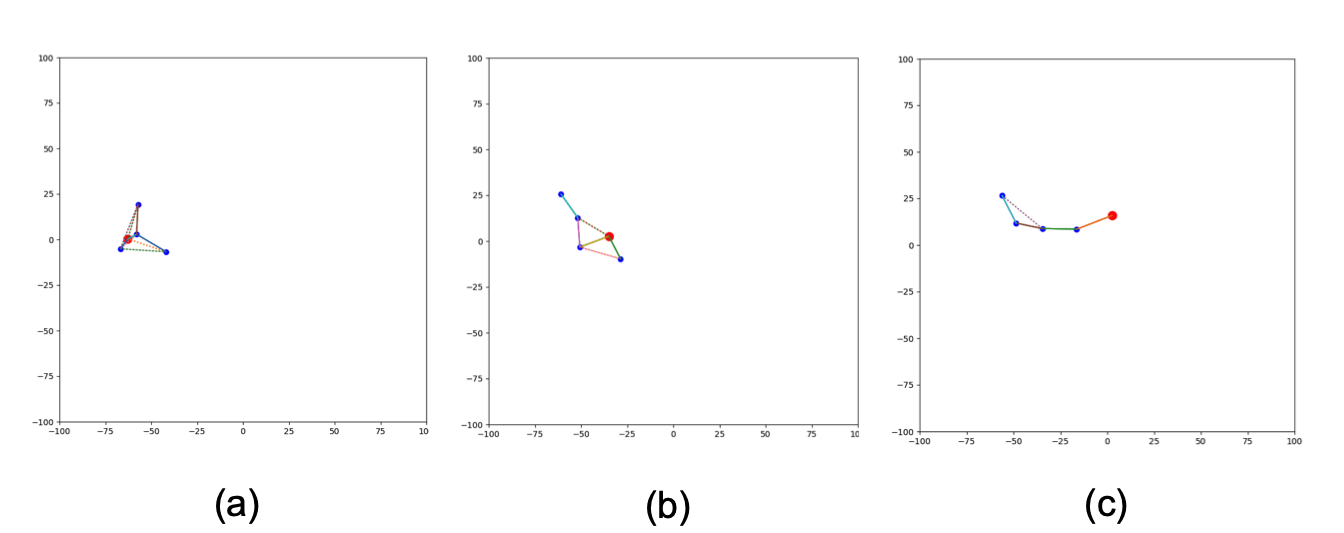}
    \caption{Following the leader. The randomly designated leader agent is shown as a large red circle, and the evolution of forming a line is presented from left to right. (a) Edges exist between every pair of agents distanced $V$ or less, and a random leader is selected (b) gradually unnecessary edges disappear while the graph remains connected (c) Final constellation: the required line structured is achieved while all the agents follow the leader agent.}
      \label{Leader}
\end{figure}

\textbf{Example $4$: Ad-hoc networking}\\

In this example, each node describes an agent with limited communication and visibility range $V$, and each solid-line edge describe an ``active" communication channel. In order to receive information from all the agents, a connected graph is required, but since the operation of transmission and reception consumes energy, there is motivation to reduce the number of necessary communication channels. In this simulation (see Figure \ref{Mesh}), all lines (both solid and dashed) are shorter than $V$, but not all the lines are necessary for maintaining connectivity, i.e. communication is guaranteed using the solid lines only. The sub graph shown in solid lines was distributedly established by the agents themselves, each agent having information only about neighboring agents in the range $V$. The upper limit of number of edges is $\frac{n(n-1)}{2}$ (i.e. clique). Under the geometric law (\ref{eq:EffectiveN}), in the $2d$ case, each agent can have up to $6$ effective agents, and therefore the upper limit of effective edges in this case will be $3n$. Toussaint \cite{toussaint1980relative} proved that finite planar set RNG, similar to ours effective graphs, contains at most $3n-6$ edges. It is clear that using the method reduces overall energy consumption associated with communication \cite{koushanfar2007techniques}.

\section*{Summary}
A general method was introduced for the distributed control of multi-agent systems, which maintains connected but flexible swarms. An algorithm was introduced, allowing to locally and in a distributed manner, to remove edges in an interconnection graph between agents, without risking disconnecting it, which reduces the number of edges in the graph over time from order $n^2$ to order $n$.

\section*{Discussion And Further Investigation}
The examples presented by simulation, did not include formal proofs. In a forthcoming article, we deal with formal definitions of motion laws in the presence of obstacles, and a formal analysis that the system indeed converges in finite expected time to the desired constellation.

An interesting extension of the ideas presented above is the design of local rules based on allowing ``at most $m$" agents of the swarm to be present in the intersection between the visibility discs of neighbours, in order to maintain the connection between them. In our previous analysis on RNG we took $m=0$. This extension named ``RNG Plus" will ensure more reliable connectivity at the expense of slightly less flexibility. This tradeoff too will be the subject of further investigation.

A video of some of the simulations shown above is available on-line at \url{https://youtu.be/A3jnYpY15DY}.

\addcontentsline{toc}{section}{References}
\bibliography{MARS_group_new}\
\bibliographystyle{unsrt}

\end{document}